\documentclass[11pt]{article}

\usepackage{verbatim}
\usepackage{amsmath}
\usepackage{amssymb}
\usepackage{amsthm}
\usepackage{yhmath}
\usepackage{graphicx}
\usepackage{subfigure}
\usepackage[usenames,dvipsnames]{color}
\usepackage[margin=1in]{geometry}

    \newtheorem{theorem}{Theorem}[section]
    \newtheorem{lemma}[theorem]{Lemma}

    \newtheorem{corollary}[theorem]{Corollary}

\usepackage{geometry}                
\usepackage{graphicx}
\usepackage{amssymb}
\usepackage{epstopdf}
\usepackage{amsfonts}
\usepackage{amssymb}
\usepackage{amsthm}
\usepackage{graphicx}
\usepackage{subfigure}

\newcommand{\e}{\varepsilon}
\newcommand{\s}{\sigma}
\renewcommand{\l}{\lambda}
\renewcommand{\O}{\Omega}
\newcommand{\C}{\mathcal{C}}
\newcommand{\T}{\mathcal{T}}

\title{Practical Conditions for Well-behaved-ness of Anisotropic Voronoi Diagrams}

\author{Guillermo D. Canas\\
CBCL-IIT, McGovern Institute.
Massachusetts Institute of Technology\\
}
\date{}                                           

\begin{document}
\maketitle

\begin{abstract}
Recently, simple conditions for well-behaved-ness of anisotropic Voronoi diagrams have been proposed. 
While these conditions ensure well-behaved-ness of two types of practical anisotropic Voronoi diagrams, as well as the geodesic-distance one, 
in any dimension, 
they are both prohibitively expensive to evaluate, and not well-suited for typical problems in approximation or optimization. 
We propose simple conditions that can be efficiently evaluated, and are better suited to practical problems of approximation and optimization. 
The practical utility of this analysis is enhanced by the fact that orphan-free anisotropic Voronoi diagrams have embedded triangulations as duals. 
\end{abstract}


\section{Introduction and Previous Work}


The Voronoi diagram, or Dirichlet tesselation, is a fundamental mathematical construct, with extensive practical application in diverse fields~\cite{Aurenhammer91}. 
Given a discrete set of sites, its definition requires a choice of distance function, which determines how the input domain is broken up into tiles of points closest to a unique site. 
The distance typically used is the natural distance in a Euclidean space, but generalizations based on other choices of distance can be defined~\cite{Power,DW,Power2,LS,LL2000}. 
The use of a particular distance can be motivated by the intended application. 
In particular, Voronoi diagrams with respect to the geodesic distance on a Riemannian manifold can be useful in defining approximations of the manifold that satisfy certain properties, 
	such as asymptotic optimality~\cite{enets,LL2000}. 

When considering applications in which specifying a Riemannian metric to induce a distance is appropriate, 
 the cost of computing geodesic paths for any pair of points may be very high for practical applications (see for instance~\cite{MMP}). 
 For this reason, fast approximations to the geodesic distance have been proposed~\cite{DW,LS}, which take constant time to evaluate, provided that the metric can be evaluated at any point in constant time. 
 However, while the geodesic-distance diagrams are guaranteed to be composed of cells, each of which contains its generating site (the diagram is orphan-free), 
 	this is, in general, not the case for the approximations of~\cite{DW,LS}. 
Since the definition of the dual of diagram with orphans is problematic (i.e.\ it may require detecting orphans and placing new sites inside each, a process that could potentially become  recursive), 
and since certain approximation or optimization guarantees may be lost in the presence of orphans, 
it is natural to consider orphan-freedom as a basic well-behave-ness condition for these anisotropic Voronoi diagrams. 
For instance, duals of two-dimensional diagrams of the type of~\cite{DW} can be shown to be embedded triangulations~\cite{adt}, if they are orphan-free. 

While the work of~\cite{DW} does not provide guarantees on orphan-freedom, that of~\cite{LS} provides an algorithm that is guaranteed to output an orphan-free diagram, in two dimensions. 
More recently, the work of~\cite{avd} provides general conditions for orphan-freedom of anisotropic Voronoi diagrams of both types, 
	which hold in any number of dimensions. The condition is that the sites form an (asymmetric) $\epsilon$-net, for a suitable value of $\epsilon$, which depends on the underlying metric. 
This condition is natural for problems in optimal covering, and $\mathcal{L}^\infty$ approximation~\cite{GruberMenets,enets}. 
However, while quite general, the above results have some drawbacks, which make them difficult to be applied in practice. In particular, for problems of average-case ($\mathcal{L}^2$) optimization or approximation, an optimal set of sites has been shown to satisfy a Delone property~\cite{enets,OQ,GruberOQ}, which is generally weaker than the net property. 
Additionally, the work of~\cite{avd} uses a Lipschitz-type constant on the underlying metric to obtain conditions for orphan-freedom. 
In practice, however, computing this constant could potentially entail visiting all pairs of points in the domain, making it prohibitively expensive to compute. 

This work addresses the above drawbacks, by providing conditions for orphan-freedom of anisotropic Voronoi diagrams that apply to more general sets of sites, 
including those typically used for $\mathcal{L}^2$ optimization and approximation, 
and use information from the metric that can be computed efficiently.

\section{Formal Setup}

Given a continuous metric $Q\in\mathcal{C}^0$ (a symmetric positive-definite matrix field) over a domain $\O\in\mathbb{R}^n$ of $n$-dimensional Euclidean space, 
	we begin by defining the functions 
	\[ D_Q^{{ }^{DW}}(a,b) = \left[ (a-b)^t Q_b (a-b) \right]^{1/2} \text{ , } a,b\in\O\]
	and
	\[ D_Q^{{ }^{LS}}(a,b) = \left[ (a-b)^t Q_a (a-b) \right]^{1/2}  \text{ , } a,b\in\O \]
which are not symmetric, but are understood as distances for the purposes of constructing Voronoi diagrams, and where, because of their asymmetry, special attention must be given to the order of their arguments. 
	
Given a discrete set $V$, the Labelle-Shewchuk diagram of $V$~\cite{LS} is composed of regions
\[ R(v) = \{ p\in\O : D_Q^{{ }^{LS}}(p,v) \le D_Q^{{ }^{LS}}(q,v), \forall q\in\O \}  \text{ , } v\in V\]
and an analogous definition of a Du/Wang diagram~\cite{DW} uses the distance $D_Q^{{ }^{DW}}$ instead. 

In the sequel, we will make extensive use of the square-root matrix field $M$, which, at each point, is the unique symmetric positive-definite matrix satisfying $M^t M=Q$. 
As shown in~\cite{avd}, $M$ must be of the same differentiability class as $Q$. 
Following~\cite{avd}, we denote the spectral matrix norm  as $\rho(\cdot)$, which, when applied to positive-definite matrices, returns the maximum eigenvalue, and 
define the function $\rho_m(\cdot)$, which returns the smallest eigenvalue of a positive-definite matrix. 

Finally, the (worst-case) metric variation constant $\s_0$ is defined as
\[ \s_0 := \displaystyle{\max_{a,b\in\O} \frac{  \rho(M_b M_a^{-1} - I) }{ \|M_a(a-b)\| } } \]
where the maximum is assumed to exist (for instance, if $\O$ is compact). 


\section{Orphan-Free Anisotropic Voronoi Diagrams of Delone Sets}

We begin by extending the results of~\cite{avd} to sets of sites that satisfy a $(C,P)$-Delone property ($C$-cover, asymmetric $P$-packing%
\footnote{$V$ is an asymmetric $P$-packing w.r.t.\ $D$ if, for all $v,w\in V$, it is either $D(v,w)\ge P$ \emph{or} $D(w,v)\ge P$.}), 
	which is a strict generalization of an asymmetric $\epsilon$-net ($\epsilon$-cover, asymmetric $\epsilon$ packing). 
The result is a simple extension of~\cite{avd}, in which the proof structure is kept exactly as is, 
	but some of the lemmas are slightly modified to account for the unequal relation between $C$ and $P$. 
We proceed by replacing $\epsilon$ by the cover constant $C$ in all the technical lemmas of~\cite{avd}, 
	except in those in which the packing property is used: Lemmas 8 and 10, which become:
	
\begin{lemma}[From Lemma 8,~\cite{avd}] 
	Let $V$ be an asymmetric $\e$-net w.r.t. $D_Q^{{ }^{DW}}$, and $v,w\in V$ be Voronoi neighbors of the resulting DW-diagram. 
	If $c\in\O$ is in the Voronoi regions of $v,w$ then
		\[ \|M_c (v-w)\| > P / (1+C\s_0) \]
\end{lemma}

\begin{lemma}[From Lemma 10,~\cite{avd}]
	Let $V$ be an asymmetric $\e$-net w.r.t. $D_Q^{{ }^{LS}}$, and $v,w\in V$ be Voronoi neighbors of the resulting LS diagram. 
	If $k=(1+C\s_0) / (1-C\s_0)$, then 
		\[ P/k \le \|M_v (v-w)\| \le C(1+k) \]
\end{lemma}

It is a somewhat tedious, but simple exercise to verify that, assuming a Delone property for the set of sites, results in the following conditions for orphan-freedom of anisotropic Voronoi diagrams. 

\begin{theorem}[Adapted from Theorem 1,~\cite{avd}]\label{th1}
Given $Q\in\C^0$ of worst-case variation $\s_0$, 
	the Du/Wang diagram of a $(C,P)$-Delone set (with respect to $D_Q^{{ }^{DW}}$) is orphan free if 
	$(P/C)^2 / \left( 2 (1 + C\s_0)^2 \right) - 2(C\s_0)^2 - 4C\s_0 > 0$.
\end{theorem}

\begin{theorem}[Adapted from Theorem 2,~\cite{avd}]\label{th2}
Given $Q\in\C^0$ of worst-case variation $\s_0$, 
	the Labelle/Shewchuk diagram of a $(C,P)$-Delone set (with respect to $D_Q^{{ }^{LS}}$) is orphan free if 
	$\left((P/C)^2 k^{-2} - \gamma^2 - 2\gamma\right)/2  - \gamma^2 - 2\gamma > 0$, where $\gamma=C\s_0(1+k)$ and $k=(1+C\s_0) / (1-C\s_0)$. 
\end{theorem}

The above are similar to the analogous theorems in~\cite{avd}, but make explicit the relation between metric variation $\s_0$, and the cover and packing constants ($C,P$). 
This generalization enables the application of these results not just to $\mathcal{L}^\infty$ optimization or approximation, 
	but also to problems in average-case ($\mathcal{L}^2$) optimization or approximation~\cite{OQ,GruberOQ}. 
Typically, a lower metric variation means that a lower cover, and higher packing constants may be used while preserving orphan-freedom. 
If we fix $\s_0$, then a lower relation $P/C$ requires a lower (stricter) cover constant to meet the well-behave-ness condition. 
For instance, for a Delone set in which $P/C < 1$ (less strict than a $C$-net, and similar to the one used in~\cite{GruberOQ}), in order to satisfy the conditions of the above theorems, 
	the cover constant must be set lower than for the net condition ($P/C=1$). 	
Finally, note that these conditions reduce to the ones in~\cite{avd} for the net case, as expected.

\section{Differentiable and Piecewise Linear Metrics}\label{C1}




The main drawback of  the above results, and those of~\cite{avd}, from a practical standpoint, is that the evaluation of the metric variation $\s_0$ may be prohibitively expensive in practice since, 
	in the worst case, it requires visiting all pair of points in the domain.  We address this problem here. 

Since $\s_0$ is a Lipschitz-type constant controlling the rate of change of $Q$, we may hope to find a simpler expression for differentiable metrics ($Q\in\C^1$) by simply taking the limit 
	in the definition of $\s_0$ as the pair of points becomes increasingly close, from every direction. In this way, we can define the differentiable-metric variation:
\[
	\s_1 := \displaystyle{\sup_{r\ne 0} \lim_{\lambda\rightarrow 0}\frac{ \rho\left( M_{p+\lambda r}  M_p^{-1} - I \right) }{ \|M_p \lambda r\| } } 
	= \displaystyle{\sup_{r\ne 0} \lim_{\lambda\rightarrow 0}\frac{ \rho\left( \frac{M_{p+\lambda r} - M_p}{\lambda} M_p^{-1}\right) }{ \|M_p r\| } } 
	= \displaystyle{\sup_{r\ne 0} \frac{\rho(D_r M_p M_p^{-1})}{ \|M_p r\| }}
\]
Clearly, the above expression can be evaluated by visiting each point in the domain only once. 

Note that, from the definitions of $\s_0$ and $\s_1$, it is
\begin{eqnarray*}
	\s_1 =  \displaystyle{\sup_{r\ne 0} \lim_{\lambda\rightarrow 0}\frac{ \rho\left( M_{p+\lambda r}  M_p^{-1} - I \right) }{ \|M_p \lambda r\| } }  \le \displaystyle{\sup_{r\ne 0} \lim_{\lambda\rightarrow 0} \s_0} = \s_0
\end{eqnarray*}

In order to apply the results of Theorems~\ref{th1} and~\ref{th2} to the case in which we only have knowledge of $\s_1$, we can simply try to obtain an upper bound of $\s_0$ using $\s_1$. 
However, while 
a differentiable function $f$ over a compact domain in $\mathbb{R}$ has a Lipschitz constant $L=\max_{p\in\O} f'_p$~\cite{rudin-analysis}, 
	there is no such simple relation of  $\s_1$ bounding $\s_0$  from above. 
The reason for this is the denominator $\|M_p r\|$ in the definition of $\s_1$, which prevents a simple analysis using integration and the Mean Value Theorem. 
It is, however, possible to prove similar results to those of Theorems~\ref{th1} and~\ref{th2} by using only the differentiable-metric variation $\s_1$, without resorting to $\s_0$.



We begin by proving the following technical lemma, whose statement is the same as Lemma 2 of~\cite{avd}, but replacing $\s_0$ by $\s_1$, 
	but whose proof requires a new analysis using different techniques, and is relegated to Appendix A in the interest of conciseness. 

\begin{lemma}\label{lem2}
	If $\s_1$ is the differentiable-metric variation  of $Q\in\C^1$, then for all $a,b\in\O$ and every constant $\e$, $\|M_a(a-b)\|\le \e$ implies
		\[ 1-\e\s_1 \le \rho_m(M_b M_a^{-1}) \le {\|M_b (a-b)\|}/{\|M_a (a-b)\|}  \le \rho(M_b M_a^{-1}) \le 1+\e\s_1 \]
\end{lemma}

Since this is the key technical lemma from which all results in~\cite{avd} stem, it can easily be verified that the above  implies well-behave-ness results using the differentiable-metric constant. 
In particular:

\begin{corollary}	\label{cor1}
Given $Q\in\C^1$ of differentiable-metric variation $\s_1$, 
	the Du/Wang diagram of a $(C,P)$-Delone set (with respect to $D_Q^{{ }^{DW}}$) is orphan free if 
	$(P/C)^2 / \left( 2 (1 + C\s_1)^2 \right) - 2(C\s_1)^2 - 4C\s_1 > 0$.
\end{corollary}

\begin{corollary}	\label{cor2}
Given $Q\in\C^1$ of differentiable-metric variation $\s_1$, 
	the Labelle/Shewchuk diagram of a $(C,P)$-Delone set (with respect to $D_Q^{{ }^{LS}}$) is orphan free if 
	$\left((P/C)^2 k^{-2} - \gamma^2 - 2\gamma\right)/2  - \gamma^2 - 2\gamma > 0$, where $\gamma=C\s_1(1+k)$ and $k=(1+C\s_1) / (1-C\s_1)$. 
\end{corollary}

In particular, if the sites form an $\e$-net, then the above results imply that a $DW$ diagram is orphan-free whenever $\s_1\e \le 0.09868$, and a $LS$ diagram is orphan-free
if $\s_1\e \le 0.0584$.  \\

\noindent{\bf Relation between metric variations $\s_0$ and $\s_1$}. 
Interestingly, it is also possible to use Lemma~\ref{lem2} above to find some form of upper bound on $\s_0$, using only knowledge of $\s_1$. 
To this end we first need to introduce the metric variation $\s_0(C)$, over neighborhoods of size $C$ as:
\[ 
	\s_0(C) := \displaystyle{ \sup_{a,b\in\O, \text{  } \|M_a(a-b)\|\le C}  \frac{ \rho(M_b M_a^{-1} - I) }{ \|M_a(a-b)\| } }
\]
which is identical to the definition of $\s_0$, but only considers pairs of points, vaguely speaking, within range $C$ of each other. 
Given Eq.~\ref{tech2} from the proof of Lemma~\ref{lem2}, the definition of $\s_1$, and the (constant velocity) parametrized straight line connecting $a,b$ ($q:[0,1]\rightarrow \overline{ab}$), it is
\begin{equation}\label{eq01}
\begin{split}
	\s_0(C) &=  \displaystyle{ \sup_{a,b\in\O, \text{  } \|M_a(a-b)\|\le C}  \frac{ \rho\left(  (M_b - M_a)M_a^{-1}   \right) }{ \|M_a(a-b)\| } } \\ 
 		 &=  \displaystyle{ \sup_{a,b\in\O, \text{  } \|M_a(a-b)\|\le C}  \frac{ \rho\left(   \int_0^1 D_{b-a}M_{q(\lambda)} M_{q(\lambda)}^{-1} M_{q(\lambda)} M_a^{-1} d\lambda  \right) }{ \|M_a(a-b)\| } } \\
		 &\le \displaystyle{ \sup_{a,b\in\O, \text{  } \|M_a(a-b)\|\le C}    \int_0^1   \frac{  \rho(D_{b-a} M_{q(\lambda)} M_{q(\lambda)}^{-1} ) }{ \|M_{q(\lambda)}(b-a)\| }   \rho( M_{q(\lambda)}M_a^{-1})^2 d\lambda }\\
		 &\le \displaystyle{ \sup_{a,b\in\O, \text{  } \|M_a(a-b)\|\le C}  \s_1  \int_0^1   \rho( M_{q(\lambda)}M_a^{-1})^2 d\lambda }\\
		 &\le \displaystyle{ \sup_{a,b\in\O, \text{  } \|M_a(a-b)\|\le C}  \s_1  \int_0^1  (1+\lambda C \s_1)^2  d\lambda } \\
		 & \le \s_1 ( 1 + C\s_1 + (C\s_1)^2/3)
\end{split}
\end{equation}
where, typically, the value of $C$ used in the above expression will be the cover constant of the sites. 

We finally note that, given $\s_1$, it is possible to use Eq.~\ref{eq01} to upper bound $\s_0$, and plug this in Theorems~\ref{th1} and~\ref{th2}. 
However, because it is always $1 + C\s_1 + (C\s_1)^2/3 > 1$,  we can always obtain better (less restrictive) conditions for well-behave-ness of anisotropic Voronoi diagrams 
by using instead  Corollaries~\ref{cor1} and~\ref{cor2} directly.

\subsection{Piecewise-Linear Metrics}\label{sec:PL}

So far we have considered conditions for well-behave-ness of anisotropic Voronoi diagrams only for metrics of class $\C^0$ and $\C^1$. 
In practice, however, it may be the case that the input metric is given as a PL function over a simplicial complex. 
This is typically the case in approximation problems in which the metric itself is derived from a previous estimate of the solution. 

In this case, the metric is almost everywhere differentiable, and therefore the integrals in Eq.~\ref{eq01}, 
	and those in the proof of Lemma~\ref{lem2} can be suitably broken into pieces in which the derivative of $M$ is defined. 
This implies that the results from Sec.~\ref{C1} apply to the PL metric case, without change. 
The definition of $\s_1$, however, becomes a supremum over all points where the metric is differentiable, which always exists in the case of PL metrics in which $M$ is 
interpolated inside simplicies, using values given at the vertices. 
In a slight abuse of notation, we replace supremums over the interior of simplicies with maximums, where it is understood that the supremum exists. 

In the piecewise-linear metric case, it is possible to make use of the simple structure of the metric to efficiently compute an upper bound of $\s_1$. 
Let $M^k_i = \partial_k M$ be the derivative of $M$ in the $k$-th coordinate direction (which is constant inside each $i$-th simplex), 
	$\T=\{\tau_i : i=1,\dots,m\}$ the set of simplicies over which $Q$ is linearly interpolated, 
	and $\{v_j : j\in I_i\}$ the set of vertices incident to the $i$-th simplex $\tau_i$.
Noticing, from its definition, that $\s_1$ is invariant to a scaling of $r$,  and from the convexity of $\rho$, it follows that:
\begin{equation}\label{eqPL}
\begin{split}
	\s_1 &= \displaystyle{ \max_{i=1,\dots,m}  \max_{p\in{\tau_i}}  \max_{\|r\|=1} \frac{\rho(D_r M_p M_p^{-1})}{ \|M_p r\| }  } \\
	&\overset{Lem. 7.1}{= }  \displaystyle{ \max_{i=1,\dots,m}  \max_{p\in{\tau_i}}  \max_{\|r\|=1}   \left[ \|M_p r\| \rho_m(M_p (D_r M)^{-1}) \right]^{-1} } \\ 
	&{= }   \displaystyle{ \max_{i=1,\dots,m}     \max_{\|r\|=1}  \left[ \min_{p\in{\tau_i}} \|M_p r\| \rho_m(M_p (D_r M)^{-1}) \right]^{-1} } \\ 
	&\le  \displaystyle{ \max_{i=1,\dots,m}     \max_{\|r\|=1}  \left[   \min_{j\in I_i} \lambda_1(v_j)^2 \rho_m( (D_r M)^{-1} ) \right]^{-1} } \\
	&\overset{Lem. 7.1}{=}  \displaystyle{ \max_{i=1,\dots,m}     \max_{\|r\|=1}  \frac{ \rho(D_r M) }{    \min_{j\in I_i} \lambda_1(v_j)^2 }  } \\
	&\le   \displaystyle{ \max_{i=1,\dots,m}     \max_{\|r\|=1}  \frac{    \sum_{k=1}^n   \rho(  M_i^k ) r_k  }{    \min_{j\in I_i} \lambda_1(v_j)^2 }  } \\
	&\le \displaystyle{ \max_{i=1,\dots,m}    \frac{ \left[   \sum_{k=1}^n   \rho(  M_i^k )^2 \right]^{1/2}  }{    \min_{j\in I_i} \lambda_1(v_j)^2 }  }
\end{split}
\end{equation}
where $\lambda_1(p)$ is the smallest eigenvalue of $M_p$, and the numerator in the sixth line has been maximized by setting
$ r^k =  \rho\left(   M_i^k  \right)  / \left[  \sum_l \rho\left( M_i^l   \right)^2 \right]^{1/2}$, which is clearly a unit vector.

The bound of Eq.~\ref{eqPL} is clearly conservative, however, crucially, it can be computed in time linear in the size of $\T$, since it is a maximum over terms, one per simplex, 
	which (for fixed dimension $n$) can be computed in constant time. 


\section{Practical Considerations}

Consider the following practical scenario. 
A metric $Q$ is given, which is PL over a simplicial complex $\T$, in $n$ dimensions (as in Sec.~\ref{sec:PL}, by linearly interpolating the square-root $M$ inside simplicies). 
We are also provided a finite set $V$ of sites. 

We consider the problem of determining whether the anisotropic Voronoi diagram of $V$ is orphan-free (whether of the Du/Wang, or Labelle/Shewchuk type), and if so, computing it efficiently. 
To answer this question, we can make use of Corollaries~\ref{cor1} and~\ref{cor2}. 
To this end, we first use Eq.~\ref{eqPL} to evaluate the differential-metric variation $\s_1$ of $Q$ (in time proportional to the size of $\T$). 
Next, we may compute the packing and cover constants of $V$ (with respect to either $D_Q^{{ }^{DW}}$ or $D_Q^{{ }^{LS}}$), and check whether  the conditions of Corollaries~\ref{cor1} and~\ref{cor2} are met, 
	in which case we are ensured that the anisotropic Voronoi diagram of $V$ is orphan-free. 
In certain cases, such as when using the simple front-propagation algorithm of~\cite{adt}, it is the case that computing the Voronoi diagram can be done efficiently (in time proportional to the size of $\T$), and the output is correct if there is a priori knowledge that the diagram is orphan-free. 

The missing step in the above reasoning is computing the covering and packing constants of $V$. 
The packing constant can be simply computed from $V$, and by evaluating the metric at points in $V$ only, since it involves only distances between sites. 
Computing, or finding an upper bound of the cover constant, however, may require computing a full anisotropic Voronoi diagram, and measuring distances from all points in each Voronoi region to their generating site. In the case of Labelle/Shewchuk diagrams, computing the diagram, as in~\cite{LS}, may itself already provide an answer to whether the diagram is orphan-free. 

As for Du/Wang diagrams, an efficient solution is to use the optimistic algorithm of~\cite{adt}. This algorithm runs efficiently (in time proportional to the size of $\T$), 
and produces a correct diagram only if the correct diagram of $V$ is orphan-free. If it isn't, then the algorithm returns a diagram that differs from the true one only at the orphans. 
If the diagram of $V$ isn't orphan-free, then the algorithm of~\cite{adt} provides an upper bound of the cover constant $C$, which cannot satisfy Corollary~\ref{cor1}, since the true diagram isn't orphan-free. If, on the other hand, the true diagram of $V$ is orphan-free, then the algorithm provides the exact cover constant, which can be used to verify the conditions of Corollary~\ref{cor1}. 
In conclusion, the use of the algorithm of~\cite{adt}, along with Corollary~\ref{cor1}, to check for orphan-freedom, never produces a false-positive, 
and only produces false-negatives in cases in which it computes the exact cover constant, and therefore the false-negative cases can all be explained only by how loose the bounds of Corollary~\ref{cor1} are, and not by the fact that the algorithm of~\cite{adt} is optimistic. We believe this point to be of practical importance. 

\section{Conclusion}

Given the problem of determining whether the anisotropic Voronoi diagram of a set of sites is orphan-free, 
we have provided an analysis that results in new conditions that can be efficiently evaluated in the case of smooth or piecewise-linear metrics and, 
perhaps more importantly, which can be applied to typical sets of sites resulting from problems in $\mathcal{L}^2$ optimization and approximation. 
While the analysis is somewhat complex and technical, the results are simple to formulate and verify. 
The emphasis throughout has been on the analysis, while the algorithmic aspects are only superficially discussed. 
We hope that this work may provide a useful piece of analysis in the ongoing progress in practical algorithms for the computation of anisotropic Voronoi diagrams and anisotropic Delaunay triangulations.

%
%
%
%
%


\newpage
\bibliographystyle{plain}
\bibliography{avdp}

\section*{Appendix A}

We prove here Lemma~\ref{lem2} which states that, given $a,b\in\O$, and $\s_1$, 
	such that $\|M_a(a-b)\| \le \e$, 
	then it is 
	 \[  1-\e\s_1 \le  \rho_m(M_b M_a^{-1}) \le {\|M_b (a-b)\|}/{\|M_a (a-b)\|}  \le \rho(M_b M_a^{-1}) \le 1+\e\s_1 \]

Clearly, it is
\[  {\|M_b (a-b)\|}/{\|M_a (a-b)\|} = \frac{\| M_b M_a^{-1} M_a (a-b) \| }{\| M_a (a-b) \|} \le \frac{\| M_b M_a^{-1}\| \| M_a (a-b) \| }{\| M_a (a-b) \|} =  \rho(M_b M_a^{-1}) \]
and a similar argument shows that 
\[  {\|M_b (a-b)\|}/{\|M_a (a-b)\|} \ge \rho_m(M_b M_a^{-1}) \]
%
%
We now show thats $1-\e\s_1 \le \rho_m(M_{b} M_a^{-1})$, and $ \rho(M_{b} M_a^{-1}) \le 1+\e\s_1$.
We begin with the following simple lemma: 

\begin{lemma}\label{lem:rhom}
	Given a non-singular matrix $A\in\mathbb{R}^{n\times n}$, it is $\rho(A^{-1}) = \rho_m(A)^{-1}$. 
\end{lemma}
\begin{proof}
	If $\l_i$, $i=1,\dots,n$ are the eigenvalues of A then
		\[ \rho(A^{-1}) = \max_i |\l_i^{-1}| = \max_i |\l_i|^{-1} = (\min_i |\l_i|)^{-1} = \rho_m(A)^{-1} \]
\end{proof}
\vspace*{0.1in}


Consider the parametrized line segment $p(\l) = b + (a-b)\cdot\l$, $\l\in[0,1]$. 
Define the unit vector $r=(a-b)/\|a-b\|$, where $\e \ge \|M_a(a-b)\| = \|a-b\| \|M_a r\|$ implies that $\|a-b\| \le \e/\|M_a r\|$. 
For any positive number $k$, the sub-multiplicative property of the spectral norm implies that a partition of $[0,1]$ into $k$ equal parts satisfies: 
\[  \rho(M_{p(\l)} M_{p(0)}^{-1}) \le \displaystyle{\prod_{i=0}^{k-1} \rho(M_{p((i+1)\Delta)} M_{p(i\Delta)}^{-1})} \]
where $\Delta = 1/k$. 

Taking logarithms on both sides, and the limit of $\Delta$ approaching $0$, 
the above inequality becomes:
\begin{eqnarray}
 \ln\rho(M_{p(\l)} M_{p(0)}^{-1}) & \le & \displaystyle{\lim_{\Delta\rightarrow 0} \sum_{i=0}^{k} \ln\rho(M_{p((i+1)\Delta)} M_{p(i\Delta)}^{-1})} \\
 					  & = & \displaystyle{\lim_{\Delta\rightarrow 0} \sum_{i=0}^{k} \ln\rho(I + \Delta \frac{M_{p((i+1)\Delta)} - M_{p(i\Delta)}}{\Delta} M_{p(i\Delta)}^{-1}) } \\
				    & = & \displaystyle{\lim_{\Delta\rightarrow 0} \sum_{i=0}^{k} \ln\left[ 1 + \Delta \rho(\frac{M_{p((i+1)\Delta)} - M_{p(i\Delta)}}{\Delta} M_{p(i\Delta)}^{-1}) \right] } \\
\label{Veq5}			    & = & \displaystyle{\lim_{\Delta\rightarrow 0} \sum_{i=0}^{k} \Delta \rho(\frac{M_{p((i+1)\Delta)} - M_{p(i\Delta)}}{\Delta} M_{p(i\Delta)}^{-1}) } \\
\label{Veq6}		         & = & \displaystyle{ \int_{0}^{\l} \rho(D_{b-a} M_{p(\tau)} M_{p(\tau)}^{-1}) d\tau } \\
				         & \le & \displaystyle{ \int_{0}^{\l} \rho(\frac{D_r M_{p(\tau)} M_{p(\tau)}^{-1}}{\|M_{p(\tau)} r\|}) \frac{\|M_{p(\tau)} r\|}{\|M_a r\|} \e d\tau } \\
				         & \le & \displaystyle{ \int_{0}^{\l} \e\s_1 \rho(M_{p(\tau)} M_{p(0)}^{-1}) d\tau}
\end{eqnarray}

Equations~\ref{Veq5} and~\ref{Veq6} require some clarification. 
Equation~\ref{Veq5} can be obtained by taking the Taylor expansion with remainder:
	$\ln(1+x)=x - c^2 x^2/2$ for some constant $c\in[0,x]$, and showing that the second term vanishes when taking the limit for $\Delta\rightarrow 0$. 

Likewise, Equation~\ref{Veq6} can be worked out backwards as:
\begin{eqnarray*}
	 \displaystyle{ \int_{0}^{\l} \rho(D_{a-b} M_{p(\tau)} M_{p(\tau)}^{-1}) d\tau } &=&  \displaystyle{ \int_{0}^{\l} \rho( \lim_{\Delta_1\rightarrow 0}\frac{M_{p(i\Delta+\Delta_1)} - M_{p(i\Delta)}}{\Delta_1}   M_{p(\tau)}^{-1}) d\tau } \\
	 				&=& \displaystyle{   \lim_{\Delta_1\rightarrow 0}  \int_{0}^{\l} \rho( \frac{M_{p(i\Delta+\Delta_1)} - M_{p(i\Delta)}}{\Delta_1}   M_{p(\tau)}^{-1}) d\tau }
\end{eqnarray*}
Since $M\in\C^1$, the limit is uniformly convergent on a compact domain $\O$, and so
we can move it outside of the integral. 
The integral exists by the continuity of its integrand, and so equals the limit over all partitions of
$[0,\l]$ of the Riemann sum as the size of the intervals of the partition go to 0.
Now the integrand depends on $\Delta_1$, but uniformly. 
Therefore, for any $\epsilon>0$ there is a $\Delta'>0$ such that, for all $\Delta_1<\Delta'$, 
	the Riemann sum associated with the equally-spaced partition of $[0,\l]$ of size $\Delta_1$ is within $\epsilon$ of the true integral value. 
Hence 
\begin{eqnarray*}
	 \displaystyle{   \lim_{\Delta_1\rightarrow 0}  \int_{0}^{\l} \rho( \frac{M_{p(i\Delta+\Delta_1)} - M_{p(i\Delta)}}{\Delta_1}   M_{p(\tau)}^{-1}) d\tau }
	 				&=&  \displaystyle{\lim_{\Delta\rightarrow 0} \sum_{i=0}^{k} \Delta \rho(\frac{M_{p((i+1)\Delta)} - M_{p(i\Delta)}}{\Delta} M_{p(i\Delta)}^{-1}) }
\end{eqnarray*}
\vspace*{0.2in}

%


The expression
\[ \ln\rho(M_{p(\l)} M_{p(0)}^{-1})  \le  \displaystyle{ \int_{0}^{\l} \e\s_1 \rho(M_{p(\tau)} M_{p(0)}^{-1}) d\tau} \]
implies that we can define an upper bound $  \rho(M_{p(\l)} M_{p(0)}^{-1}) \le \mu(\l)$ such that:
\[  \ln\mu(\l)=\displaystyle{ \int_{0}^{\l} \e\s_1 \mu(\l) d\tau}  \] 
which, taking derivatives, becomes: 
\begin{equation}\label{Veq:BODE}
	\mu^{-1}(\l) \frac{d \mu(\l)}{d\l} = \e\s_1 \mu(\l)
\end{equation}

Equation~\ref{Veq:BODE} is a Bernoulli ordinary differential equation. 
Given the initial condition $\mu(0) = \rho(M_b M_b^{-1}) = 1$, it can be integrated it to obtain
$\mu(\l) = (1 - \e\s_1\l)^{-1}$. 
Since $\mu(\l)$ is an upper bound of $\rho(M_{p(\l)} M_{p(0)}^{-1})$, this implies
\[    \rho(M_{p(\l)} M_{p(0)}^{-1}) \le (1 - \e\s_1\l)^{-1} \]
and, since $a=p(1)$ and $b=p(0)$, in particular
\[    \rho(M_{a} M_b^{-1}) \le (1 - \e\s_1)^{-1} \]

By Lemma~\ref{lem:rhom} it is
\[ \rho_m(M_b M_a^{-1})^{-1} =  \rho(M_{a} M_b^{-1})^{-1} \le (1 - \e\s_1)^{-1} \]
and thus
\[  \rho_m(M_b M_a^{-1}) \ge 1 - \e\s_1 \]
as claimed.

The other inequality $\rho(M_b M_a^{-1}) \le 1+\e\s_1$ can be obtained with a similar argument
	by using the identity $\rho_m(A\cdot B) \ge \rho_m(A) \rho_m(B)$, and the fact that $\rho_m(A)\ge 1 - \rho(A-I)$:
\begin{eqnarray*}
 \ln\rho_m(M_{p(\l)} M_{p(0)}^{-1}) & \ge & \displaystyle{\lim_{\Delta\rightarrow 0} \sum_{i=0}^{k} \ln\rho_m(M_{p((i+1)\Delta)} M_{p(i\Delta)}^{-1})} \\
 					  & = & \displaystyle{\lim_{\Delta\rightarrow 0} \sum_{i=0}^{k} \ln\left[1 - \rho(\Delta \frac{M_{p((i+1)\Delta)} - M_{p(i\Delta)}}{\Delta} M_{p(i\Delta)}^{-1})\right] } \\
				    & = & -\displaystyle{\lim_{\Delta\rightarrow 0} \sum_{i=0}^{k} \Delta \rho(\frac{M_{p((i+1)\Delta)} - M_{p(i\Delta)}}{\Delta} M_{p(i\Delta)}^{-1}) } \\
				         & = & -\displaystyle{ \int_{0}^{\l} \rho(D_{a-b} M_{p(\tau)} M_{p(\tau)}^{-1}) d\tau } \\
				         & = & -\displaystyle{ \int_{0}^{\l} \rho(\frac{D_r M_{p(\tau)} M_{p(\tau)}^{-1}}{\|M_{p(\tau)} r\|}) \frac{\|M_{p(\tau)} r\|}{\|M_a r\|} \e d\tau } \\
				         & \ge &- \displaystyle{ \int_{0}^{\l} \e\s_1 \rho(M_{p(\l)} M_{p(0)}^{-1}) d\tau}
\end{eqnarray*}
where the same technical considerations as those discussed for Equations~\ref{Veq5} and~\ref{Veq6} apply here as well. 

A lower bound is $\mu_m(\l) \le \rho_m(M_{p(\l)} M_{p(0)}^{-1})$ with
\[ \ln\mu_m(\l)= -\displaystyle{ \int_{0}^{\l} \e\s_1 \mu_m(\l) d\tau}  \] 
which, taking derivatives, and integrating a differential equation similar to Equation~\ref{Veq:BODE}, with initial condition $\mu_m(0) = \rho_m(I)=1$, results in $\mu_m(\l) = (1 + \e\s_1\l)^{-1}$
and thus
\[  \rho_m(M_a M_b^{-1}) = \rho_m(M_{p(1)} M_{p(0)}^{-1}) \ge (1 + \e\s_1)^{-1} \]
which, using Lemma~\ref{lem:rhom} implies
\[ \rho(M_b M_a^{-1})^{-1} =  \rho_m(M_a M_b^{-1}) \ge (1 + \e\s_1)^{-1} \]
and thus
\begin{equation}\label{tech2}
 \rho(M_b M_a^{-1}) \le 1 + \e\s_1 
\end{equation} 
as claimed.

\end{document}